\documentclass[%
reprint,
amsmath,amssymb,
]{revtex4-1}

\usepackage{graphicx}
\usepackage{subfig}
\usepackage{dcolumn}
\usepackage{bm}

\usepackage{amsthm}
\newtheorem*{theorem}{Theorem}

\begin{document}

\title{Linear scaling conventional Fock matrix calculation \\ with stored non-zero two-electron integrals}

\author{Alexander V. Mitin}
\affiliation{Moscow Institute of Physics and Technology,  9 Institutskiy per., Dolgoprudny, Moscow Region, 141701, Russia;} 
\affiliation{Joint Institute for High Temperatures of RAS, Izhorskaya st. 13 Bd.2, 125412 Moscow, Russia}
\email{$mitin.av@mipt.ru$}
\date{\today}
\begin{abstract}
It was shown that the conventional Fock matrix calculation, which using stored non-zero two-electron integrals, density prescreening, and data compression methods possesses the linear scaling property with respect to the problem size. This result follows from the proven theorem, which reads that the total number of non-zero two-electron integrals scales asymptotically linearly with respect to the number of basis functions for large molecular systems. An analysis of the Fock matrix calculation with density and density difference prescreening shows that its linear scaling property arises due to the asymptotically linear scaling properties of the number of non-zero two-electron integrals and the linear scaling property of the number of leading matrix elements of the density matrix. The use of the density and the density difference prescreening in the Fock matrix calculation consequently enforces this property. The conventional Fock matrix calculation with the stored non-zero two-electron integrals has been reformulated to the Fock matrix calculation with density or density difference prescreening by using the data compression technique to store the non-zero two-electron integrals and their indices. The numerical calculations with this method show that its linear scaling property begins from 2500 to 4000 basis functions in dependence on the basis function type in molecular calculations.
\end{abstract}
\maketitle

\section{Introduction}

It is well known that the Fock matrix calculation is one of the most computationally expansive step in the Hartree-Fock and the density functional theory methods \cite{Hatree_ProcCambridgePhilSoc_128_24_89,Fock_ZPhys_1930_61_126, HohenbergKohn_PR_1964_136_B864,KohnSham_PR_1965_140_A1133}. Its efficient realization is mainly defined by explicit inclusion of the main hardware features of the used computer into developed numerical algorithm and program code. The main characteristic of computers since 1950th was significantly faster CPU in comparison with input/output (I/O) capabilities. For this type of computers have been developed efficient numerical algorithms of the direct Fock matrix assembly by calculating only Fock matrix changes for known density difference with subsequent construction of the full Fock matrix from these increments \cite{AFK_JCC_1982_3_385, WhiteHeadGordon_JCP_1994_101_6593, SSF_Science_1996_271_51, CSA_JCP_1996_104_4685}.

However, now, current trend in computer architecture consists in subsequent significant improvement of its I/O capabilities in comparison with growing of its CPU performance. This changes appears due to solid state drive (SSD) technology. This is log term trend because the limits of this technology is too far to be exhaust. For example, the third generation of Microsemi Adaptec SmartRAID controllers permits to build a disc system with I/O rate up to 6.6 GB/sec, while the advanced Seagate Nytro XP7200 PCIE add-in SSD has sequential sustained read performance equals to 10000 MB/s at capacity 7.7 TB. The efficient data compression technique \cite{Mitin_Theochem_2002_115} permits to keep non-zero two-electron integrals mainly in one byte words. This means that the modern external disk systems can read the non-zero two-electron integrals from the external disk system to the main computer memory with rate from $6\times10^9$ to $10\times10^9$ integrals per second. On the other hand, such high speed of two-electron integrals calculations has not yet been demonstrated on single/double socket computers.

Therefore, a critical analysis and reinvestigation of the conventional method Fock matrix construction from the non-zero two-electron integral stored on an external disk system is an important problem, which can open new perspectives for a development of an efficient algorithm of Fock matrix calculation on computers with extended I/O capabilities.  

In this connection, the scaling property of the total number of non-zero two-electron integrals was investigated and presented in Section \ref{Section_2}. Then, a theoretical analysis of the Fock matrix calculation with density and density difference prescreening are given in Section \ref{Section_3}. After that, the conventional Fock matrix calculation from the stored non-zero two-electron integrals with the density difference prescreening and the data compression was formulated and its linear scaling property was demonstrated in numerical calculations. These results are presented in Section \ref{Section_4}.

\section{Linear scaling property of non-zero two-electron integrals} \label{Section_2}

It is well known that the direct Fock matrix calculation \cite{AFK_JCC_1982_3_385}, which uses the density difference for an evaluation of the two-electron integrals with the fast multipole techniques \cite{WhiteHeadGordon_JCP_1994_101_6593, SSF_Science_1996_271_51} and the quantum chemical tree code \cite{CSA_JCP_1996_104_4685}, super-linear scales with respect to the problem size. The scalability of such a calculation is noticeably lower than quadratic with respect to the problem size and in some cases it is close to linear dependence. The explanations of this phenomenon were also given in these publications. They are based on an numerical estimation of scaling property of the total number of non-zero two-electron integrals. Formally the total number of two-electron integrals scales as $N^4$ for a system with $N$ basis functions. But, only non-zero integrals are important in real calculations. In this connection, in \cite{Dyczmons_TCA_1973_28_307} it was argued that the total number of non-zero two-electron integrals ($N_{2e}$) scales as $N^2(lnN)^2$. Later, without a proof, it was noted \cite{Ahlrics_TCA_1974_33_157} that $N_{2e}$ scales as $N^2lnN$. An estimation that $N_{2e}$ scales as $N^{2.2-2.3}$ for integrals with absolute values greater $10^{-10}$, which has been reported in \cite{StroutSciseria_JCP_1995_102_8448}, is now widely accepted. 

On the other hand, the small and large two-electron integrals are used differently in the modern direct algorithm of the Fock matrix calculations \cite{AFK_JCC_1982_3_385}. Only integrals, whose contributions to the Fock matrix are greater than predefined threshold, are calculated and used in calculations of the Fock matrix. This means that the integrals, which are small in absolute values, are calculated really, while those one, which absolute values are large, calculated often. Therefore, the scaling properties of the total numbers of non-zero small and large two-electron integrals have to be estimated separately first before investigation of the scaling property of the Fock matrix calculation on the number of basis functions. 

Possible scaling property the number of non-zero two-electron integrals can be qualitatively understand by considering a classification of the two-electron integrals on four classes in accordance with the number of centers together with estimations of numbers of integrals in these classes presented in Table \ref{Table_1}. 
\begin{table} [h]
	\caption{Four classes of two-electron integrals, number of integrals in each class, and asymptotic contributions of the integrals to the number of non-zero integrals.}
	\begin{ruledtabular}
		\begin{tabular}{@{}cccc}
			Int. class & Numb. of int. & Asympt. contr. to non-zero int. \\
			\colrule
			1c int.  &   $\sim N$    &         Yes            \\
			2c int.  &   $\sim N^2$  &         No             \\
			3c int.  &   $\sim N^3$  &         No             \\
			4c int.  &   $\sim N^4$  &         No             \\
		\end{tabular}
	\end{ruledtabular}
	\label{Table_1}
\end{table}
From this classification follows that $N_{2e}$ can only have one of the following possible asymptotic dependencies on $N$: the first, the second, the third, and the fourth power, because this dependence is defined by the asymptotic property of integrals, belonging to the corresponding classes, with respect to increasing the distances between their centers. Then, the results presented in work \cite{StroutSciseria_JCP_1995_102_8448} show that at least the four-center integrals give no contribution to the total number of non-zero integrals. Otherwise, any non-zero part of four-center integrals has to give the fourth power dependence of the number of non-zero integrals due to highest power dependence of this class in comparison with others. The complete elimination of these integrals arises because an asymptotic value of any four-center two-electron integral with respect to increasing the distances between their centers is equal to zero. (Detailed consideration of asymptotic properties of the two-electron integrals is given below.) The asymptotic property of the two- and three-center integrals with respect to the distances between the centers is similar to that of the four-center integrals. Therefore, they also must have no influence on the asymptotic property of the number of non-zero two-electron integrals. The only what they can do is shifting the linear scaling behavior of the number of non-zero two-electron integrals towards to larger molecular systems. The one-center two-electron integrals do not depend on the distances between the centers. Therefore, their values are constants and, hence, they give asymptotically linear scaling dependence of the total number of non-zero two-electron integrals.

For this reason, the fractional asymptotic values of the power dependence of the number of non-zero two-electron integrals from the number of basis functions given in \cite{StroutSciseria_JCP_1995_102_8448} obviously point out that the correct asymptotic values of these quantities have not been reached. 

Now, consider the two-electron integral
$$
(kl\vert mn)=\int\phi^a_k(1)\phi^b_l(1)\frac{1}{r_{12}}\phi^c_m(2)\phi^d_n(2)d\tau_1\tau_2
$$
where $\phi^a_k(1)$ is the $k$-th basis function located at the $a$ nucleus of the first electron and $r_{12}$ is the distance between two electrons. To prove the mathematical statement about power dependence of $N_{2e}$ on $N$ and, hence, on the size of a system, the two-electron integrals must be preliminary classified on the numbers of centers, and the dependencies of the values of these integrals on the distances between their centers in each class must be considered. After that, the theorem on power dependence of $N_{2e}$ on $N$ can be proven. 

Let us define a model system which is formed by the same atoms uniformly distributed in three-dimensional space with one $s$ basis function per each nucleus and which will be further used in the present consideration. 

Altogether, there are four types of two-electron integrals: one-, two-, three-, and four-center integrals.

{\it \bf One-center two-electron integrals}. These integrals do not depend on the distances between centers because all basis functions are located on the same center. Therefore, the values of these integrals are constants and are defined by the parameters of basis functions. There are only $N$ such integrals for the model system under consideration.  

{\it \bf Two-center two-electron integrals}. There are two sub-types of these integrals. In the first one 
$$
\int\phi^a_k(1)\phi^a_k(1)\frac{1}{r_{12}}\phi^b_l(2)\phi^b_l(2)d\tau_1\tau_2
$$
the basis functions of each electron are located on the same center. Therefore, this integral is the two-electron integral for two density distributions. It is clear that the value of such an integral decreases with increasing the distance between the two centers. Probably, it is not possible derive an explicit formula for the value of such an integral. Though, for the distances between two centers which are larger than the sum of typical sizes of these distributions, an asymptotic formula for the value of this integral can be derived if it is noted that the integral in this case is a Coulomb potential energy between two density distributions. In Electrostatics it is well known that for distances larger than the radius of the charge density distribution, the distribution can be replaced by a corresponding point charge located in the center of the density distribution \cite{Pursell_1984}. For Gaussian function, due to symmetry, such a center is a point where the function is located. The integral, in this case, is expressed by a simple formula
$$
\int\phi^a_k(1)\phi^a_k(1)\frac{1}{r_{12}}\phi^b_l(2)\phi^b_l(2)d\tau_1\tau_2\:\:\sim \:\:\frac{{N^a_k}^2 {N^b_l}^2}{R_{ab}},
$$
where $R_{ab}$ is the distance between the centers $a$ and $b$ and $N^a_k$, $N^b_l$ are normalization factors of the corresponding basis functions. 

The values of this two-center two-electron integral for some distances between two centers and the inverse values of these distances for cases of two $s$ and two $p$ Gaussian functions with exponents equal to 0.5 are presented in Table \ref{Table_2}. The integrals were calculated using the Rys polynomial method \cite{KingDupuis_JCompPhys_1976_21_144}. An analysis of this data shows that the asymptotic formula gives reasonable estimations of two-center two-electron integral values, which are especially good for the case of $s$ functions. In this case, the asymptotic and numerical values are in excellent agreement for distances larger than 2.0 \AA. 

\begin{table*}
\caption{The values of two-center two-electron integrals for different distances between the centers.}
\begin{ruledtabular}
\begin{tabular}{@{}rccc}
 R(\r{A})&    $(s^as^a|s^bs^b)$  &$(p^a_x p^a_x|p^b_x p^b_x)$&    1/$R_{ab}$ (a.u.) \\
\colrule 
     0.6 &  6.554181989463915E-1  &  5.482973740079550E-1   &  8.819620820666667E-1  \\
     1.0 &  4.980644757727490E-1  &  4.322272962735285E-1   &  5.291772492400000E-1  \\
     2.0 &  2.645470381932461E-1  &  2.489481693180657E-1   &  2.645886246200000E-1  \\
     4.0 &  1.322943123099960E-1  &  1.300701033377158E-1   &  1.322943123100000E-1  \\
     6.0 &  8.819620820666761E-2  &  8.752217465970426E-2   &  8.819620820666667E-2  \\
    60.0 &  8.819620820666759E-3  &  8.818934900256092E-3   &  8.819620820666667E-3  \\
   100.0 &  5.291772492400056E-3  &  5.291624316992965E-3   &  5.291772492400000E-3  \\
   600.0 &  8.819620820666760E-4  &  8.819613960273960E-4   &  8.819620820666667E-4  \\
  6000.0 &  8.819620820666761E-5  &  8.819620752062707E-5   &  8.819620820666667E-5  \\
 10000.0 &  5.291772492400056E-5  &  5.291772477581577E-5   &  5.291772492400000E-5  \\
 60000.0 &  8.819620820666761E-6  &  8.819620819980713E-6   &  8.819620820666667E-6  \\
\end{tabular}
\end{ruledtabular}
\label{Table_2}
\end{table*}

In the second two-center two-electron integral  
$$
\int\phi^a_k(1)\phi^b_l(1)\frac{1}{r_{12}}\phi^a_k(2)\phi^b_l(2)d\tau_1\tau_2
$$ 
the basis functions of each electron are located on the different centers, however, the centers are the same for each pair. The use of the Gaussian product theorem \cite{Shavitt_MethCompPhys_1963_2_1} transforms this integral into the one-center two-electron integral considered above with a prefactor which exponentially decreases with increase of the distance between the centers
$$
exp\left(-\frac{2\alpha\beta}{\alpha+\beta}R_{ab}^2\right) \int\phi^e_p(1)\frac{1}{r_{12}}\phi^e_p(2)d\tau_1\tau_2.
$$
Here, $\alpha$ and $\beta$ are exponents of Gaussian functions $\phi^a_k$ and $\phi^b_l$, and $\phi^e_p$ are Gaussian functions with exponent $\alpha+\beta$ located at a center $e$ between centers $a$ and $b$. 

Thus, the values of both types of two-center two-electron integrals decrease with increasing distances between the centers and, therefore, their asymptotic values are equal to zero in the limit of infinite distances between the centers.

{\it \bf Three-center two-electron integrals} 
$$
\int\phi^a_k(1)\phi^b_l(1)\frac{1}{r_{12}}\phi^c_m(2)\phi^c_m(2)d\tau_1\tau_2. 
$$
The Gaussian product theorem can be applied, in this case, to a couple of functions of the same electron located on different centers. In such a way, the tree-center two-electron integral can be transformed into a two-center two-electron integral with an exponential prefactor depending on the distance between the centers
$$
exp\left(-\frac{\alpha\beta}{\alpha+\beta}R_{ab}^2\right) 
\int\phi^e_p(1)\frac{1}{r_{12}}\phi^c_m(2)\phi^c_m(2)d\tau_1\tau_2 .
$$
The two-center two-electron integral itself is exactly the integral which has been considered above. The value of this integral depends on the distance between $e$ and $c$ centers and decreases with its increase. The value of this integral also exponentially decreases with increasing the distance between centers $a$ and $b$. Hence, the asymptotic value of this integral is equal to zero.

{\it \bf Four-center two-electron integrals}. The double application of the Gaussian product theorem transforms this integral into a two-center two-electron integral with two prefactors exponentially depend on distances between the centers
\begin{equation}
\begin{split}
exp\left(-\frac{\alpha\beta}{\alpha+\beta}R_{ab}^2\right) 
exp\left(-\frac{\gamma\delta}{\gamma+\delta}R_{cd}^2\right)  \nonumber \\
\times \int\phi^e_p(1)\frac{1}{r_{12}}\phi^f_q(2)d\tau_1\tau_2 .
\end{split}
\end{equation}
All three multipliers in this expression decrease with increasing distances between the centers. Therefore, the asymptotic value of this integral is zero.

Thus, only the one-center two-electron integrals do not depend on geometry and the number of these integrals is proportionally to $N$. The values of all other integrals depend on the distances between centers and their asymptotic values are equal to zero. Therefore, in large molecules, where distances between the majority of atoms or the centers of two-electron integrals are large, the values of these integrals will be negligibly small and, hence, they will have no influence on the asymptotic property of the number of non-zero integrals. Thus, the following Theorem can be formulated: 

\begin{theorem}
The number of two-electron integrals $N_{2e}(\epsilon)$, whose absolute values greater than $\epsilon$, linearly scales on $N$ for large molecular systems.
\end{theorem}

\begin{proof}
Consider the number of two-electron integrals for the model system, which can be presented in the following way
\begin{eqnarray}
N_{2e}(\epsilon)=\sum^{N}_{k,l,m,n=1}I(klmn)=\sum^{N}_{k=1}\sum^{N}_{l,m,n=1}I(klmn) \;,
\label{l01}
\end{eqnarray}
where \mbox{$I(klmn)=1$}, when \mbox{$\vert(kl\vert mn)\vert\geq\epsilon$}, and \mbox{$I(klmn)=0$} otherwise; and indices $k,l,m,n$ are ordered in usual way \mbox{$k \ge l$}, \mbox{$m \ge n$}, and \mbox{$k > m$}; if \mbox{$k = m$} then \mbox{$l \ge n$}. Define a tree-dimension manifold $S_a^k$ of radius $R_a^k$, whose centrum is located on the $a$ nucleus. Then, for any small $\epsilon$ it is possible define such an $R_a^k$ that, when \mbox{$b, c, d \in S_a^k$}, then \mbox{$\vert(kl\vert mn)\vert\geq \epsilon$} and \mbox{$\vert(kl\vert mn)\vert<\epsilon$} otherwise. (The possibility of defining such $R_a^k$ follows from the consideration of two-electron integrals given above, where it was shown that the values of all two-, three-, and four-center two-electron integrals asymptotically decrease to zero with increasing the distances between the centers.) Now, we can see that summation for fixed $k$ in the last triple sum in (\ref{l01}) have to be done not over the full range of indices from 1 to $N$, but only for parts of them, which are defined by manifolds $\{S_a^k\}_{k=1}^N$. The radius $R_a^k$ of manifold $S_a^k$ only depends on $\epsilon$, but not on $N$. Therefore, the triple summation for fixed $k$ gives a constant which is defines by $\epsilon$, but does not depend $N$. Summation over all indices, from 1 to $N$, occurs only when $\epsilon=0$ and, hence, $R_a^k=\infty$. Thus, the quadruple summation in (\ref{l01}) reduces to a single summation of some constants depending on $\epsilon$. This shows that the number of non-zero two-electron integrals scales linearly with respect to a number of basis functions when $\{R_a^k\}_{k=1}^N$ are small in comparison to the typical size of a large system. In the opposite case, a deviation from linear scaling will be observed for $N_{2e}(\epsilon)$. 
\end{proof} 

Before considering the results of numerical examples, it needs to be noted that the values of two-center two-electron integrals of the first type only reduce as $1/R$. Therefore, these integrals can be excluded from calculations only at extremely large distances between the centers. However, a number of these integrals formally scales only as $N^2$ in comparison with dependence $N^4$ for a total number of two-electron integrals. Hence, a contribution of these integrals to the total number of two-center integrals decreases with increase of the number of basis functions. For example, in the largest calculations presented below, with 31409 basis functions, the formal weight of these integrals in the total number of two-electron integrals is about of $1.0*10^{-9}$. This small contribution permits to conclude that the influence of these integrals on the scaling property of the total number of non-zero integrals will be practically invisible, even when taking into account that the weight of these integrals in the total number of non-zero integrals can increase with $N$ due to slower asymptotic convergence of them in comparison with other integrals.

It can be also noted that due to the formal $N^2$ dependency all these integrals are calculated during a time comparable with that one for a calculation of one-electron integrals. Therefore, a calculation of these integrals has no significant influence on the total time of two-electron integral calculation.

The numerical investigations of differential properties of $N_{2e}(\epsilon)$ require that the size of molecular systems increases with a constant increment, because only in this case the finite difference derivatives of $N_{2e}(\epsilon)$ evaluated at different $N$ can be correctly compared. This requirement can be realized only in one-, two-, and quasi three-dimensional molecular systems. The variations of $\epsilon$ or $R$ in numerical experiments can be easy simulated by counting $N_{2e}(\epsilon)$ for different values of two-electron integral cutoff thresholds.

In the present study, the $N_{2e}(\epsilon)$ numbers were counted for different two-electron integral cutoff thresholds for three systems. The first one was a linear chain of lithium atoms, where the number of atoms are varied from 10 to 1500. The second case was a two-dimensional $40\times40$ cell of lithium atoms, while the last one was alanine polymers, where the number of alanine blocks varied from 10 to 790 and the number of atoms varied from 112 to 7912. The internuclear distances in the first two examples correspond to the distances in the solid lithium. The geometries of alanine polymers were optimized by the Tinker program \cite{RenPonder_JPCB_2003_107_5933}. The STO-3G basis \cite{HSP_JCP_1969_51_2657} was employed in the calculations of the first two systems and the 6-31G basis \cite{HDP_JCP_1972_56_2257} was used in the calculations of the alanine polymers. The number of contracted Gaussian type basis functions (CGTF) in the calculations of these systems varied from 50 to 7500, from 200 to 8000, and from 609 to 43509 correspondingly. The calculations have been performed using the program \cite{Mitin_JCC_1998_19_1877} with an improved integral code.

The dependencies of $N_{2e}(\epsilon)$ on the number of basis functions for the different cutoff thresholds obtained for the first two systems in these calculations are presented on Fig. \ref{Figure_01_02}.
\begin{figure}[ht]
	\centering
	\subfloat[]{\label{fig_01} 
		\includegraphics[scale=0.63]{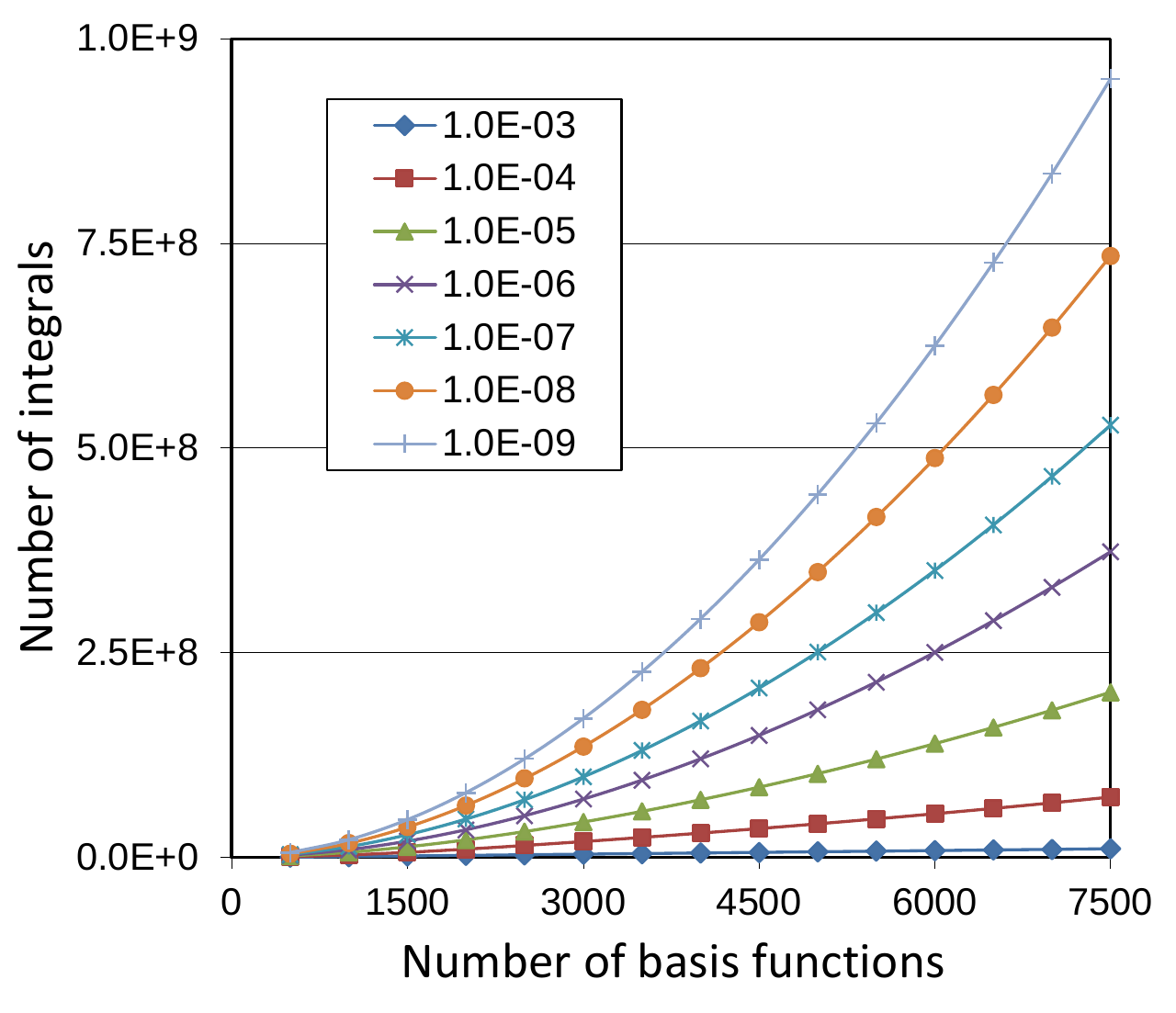}
	} 
    \\
	\subfloat[]{\label{fig_02}
		\includegraphics[scale=0.63]{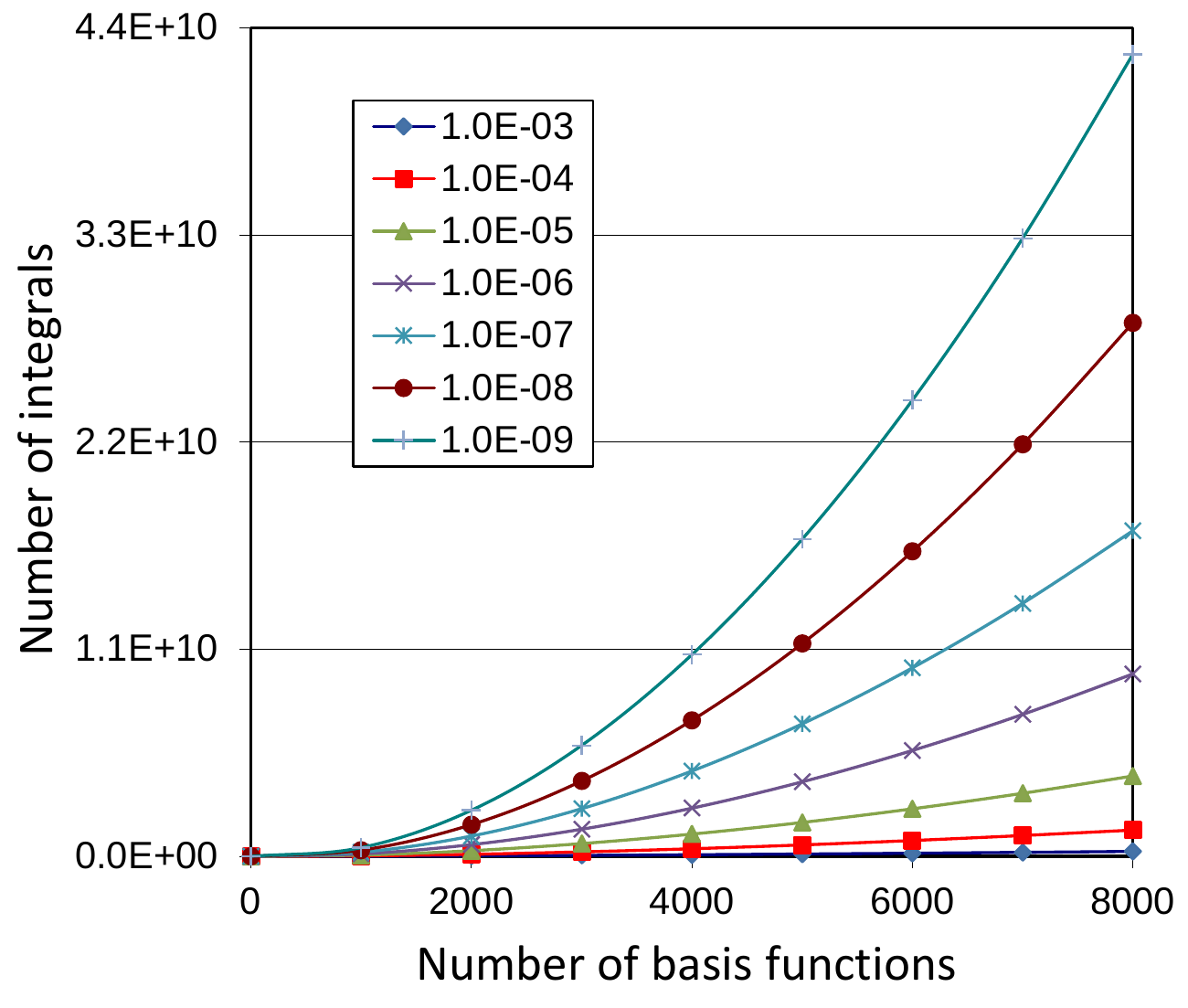}
	} 
	\caption{The dependence of a number of non-zero two-electron integrals on a number of basis functions in calculations of a linear chain of Li atoms (a) and two-dimensional $40\times40$ cell of Li atoms (b) with STO-3G basis for different cutoff thresholds.}
	\label{Figure_01_02}
\end{figure}
The dependencies of $N_{2e}(\epsilon)$ from $N$ on these figures are quite similar for both molecular systems. In the third case the dependence of $N_{2e}(\epsilon)$ on $N$ for low precision of integrals has been investigated in detail because the dependence of these integrals on $N$ defines the scaling property of the Fock matrix calculation method with the use of density or density difference prescreening. Indeed, in modern programs the two-electron integrals are calculated when an estimation of maximal integral $max(ij|kl)$ in a block of shell integrals $|max(ij|kl)\times max\Delta D_{kl}|>\epsilon$, where $max\Delta D_{kl}$ is the largest matrix element of the density matrix difference. The values of $\epsilon$ usually equals between $10^{-6}$ and $10^{-7}$ in calculations of large molecular systems \cite{ChallacombeSchwegler_JCP_1997_106_5526}. For such systems, the convergence criterion on density matrix difference has the same value. Therefore, large non-zero two-electron integrals defines the scaling property of the Fock matrix calculations.

\begin{figure}[ht]
	\centering
	\subfloat[]{\label{fig_03} 
		\includegraphics[scale=0.67]{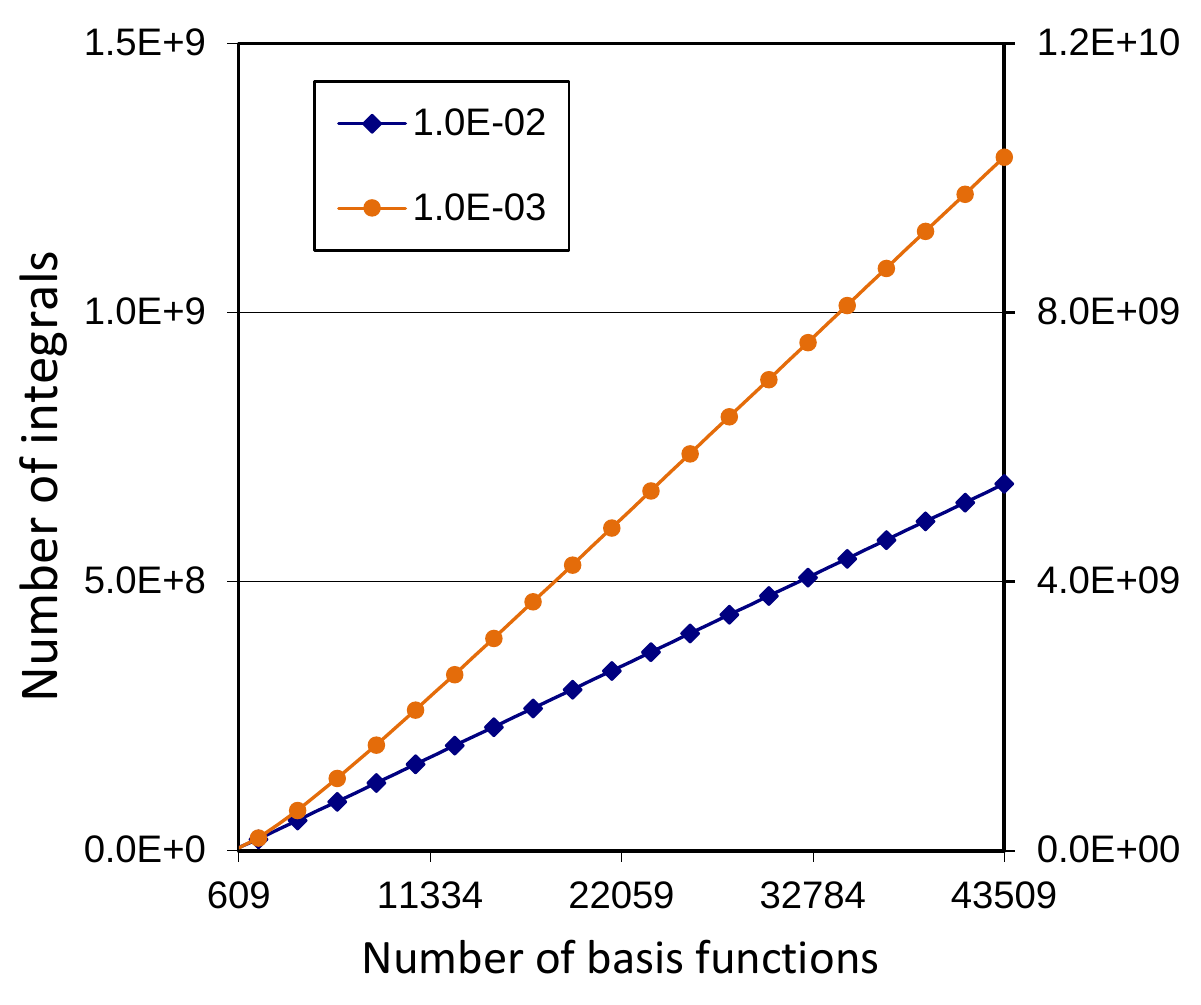}
	}
    \\ 
	\subfloat[]{\label{fig_04}
		\includegraphics[scale=0.67]{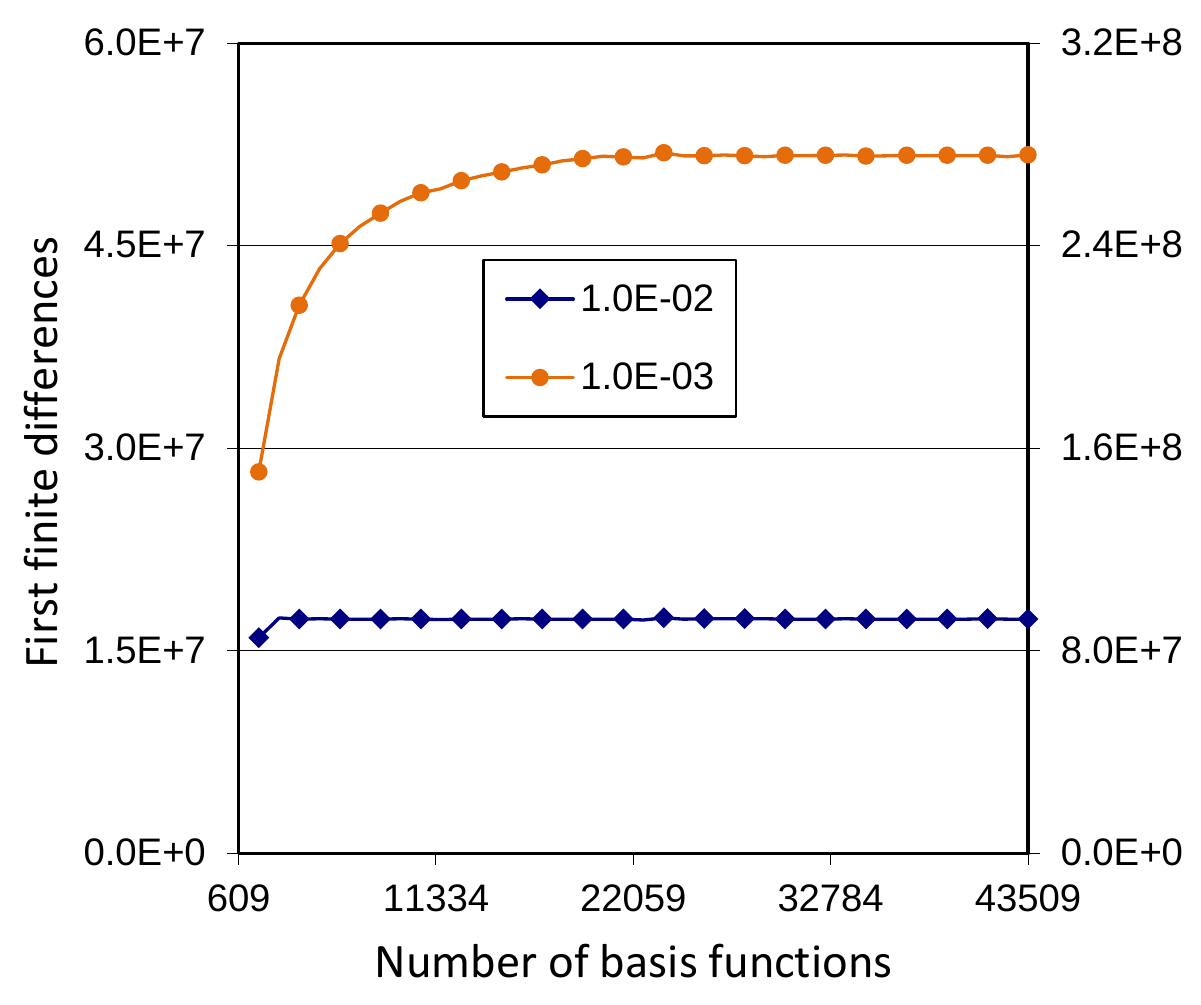}
	} 
	\caption{Dependence of $N_{2e}(\epsilon)$ (a) and first finite differences (b) on $N$ in calculations of alanine polymers with 6-31G basis for integral precision $10^{-2}$ and $10^{-3}$. The left and the right vertical scales correspond to precision $10^{-2}$ and $10^{-3}$.}
	\label{Figure_03_04}
\end{figure}

The dependencies $N_{2e}(\epsilon)$ and its first finite differences on $N$ are presented on different figures due to significant deviations in number of non-zero integrals. These dependencies for $\epsilon$ equals to $10^{-2}$ and $10^{-3}$ are presented on Fig. \ref{fig_03} and Fig. \ref{fig_04} correspondingly.

Similar dependencies for integral cutoff precision of $5\times10^{-4}$, $3\times10^{-4}$, $2\times10^{-4}$, $1.5\times10^{-4}$ and $10^{-4}$ are presented on Fig. \ref{fig_05} and Fig. \ref{fig_06}, while dependencies for precision of $10^{-5}$, $10^{-6}$ and $10^{-7}$, $10^{-8}$ are given on Fig. \ref{Figure_07_08} and Fig. \ref{Figure_09_10}, correspondingly.
\begin{figure}[h]
	\centering
	\subfloat[]{\label{fig_05} 
		\includegraphics[scale=0.67]{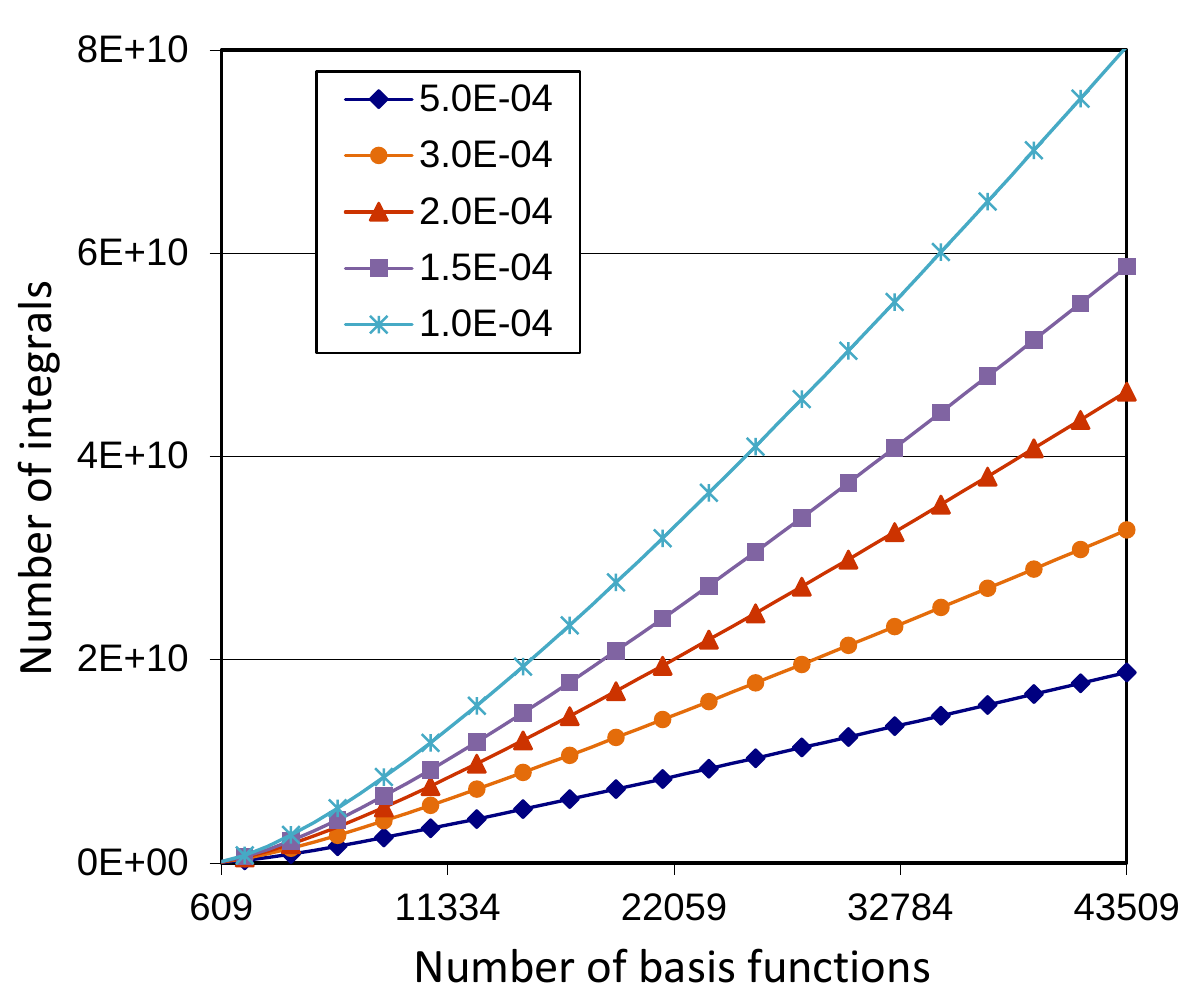}
	}
    \\ 
	\subfloat[]{\label{fig_06}
		\includegraphics[scale=0.67]{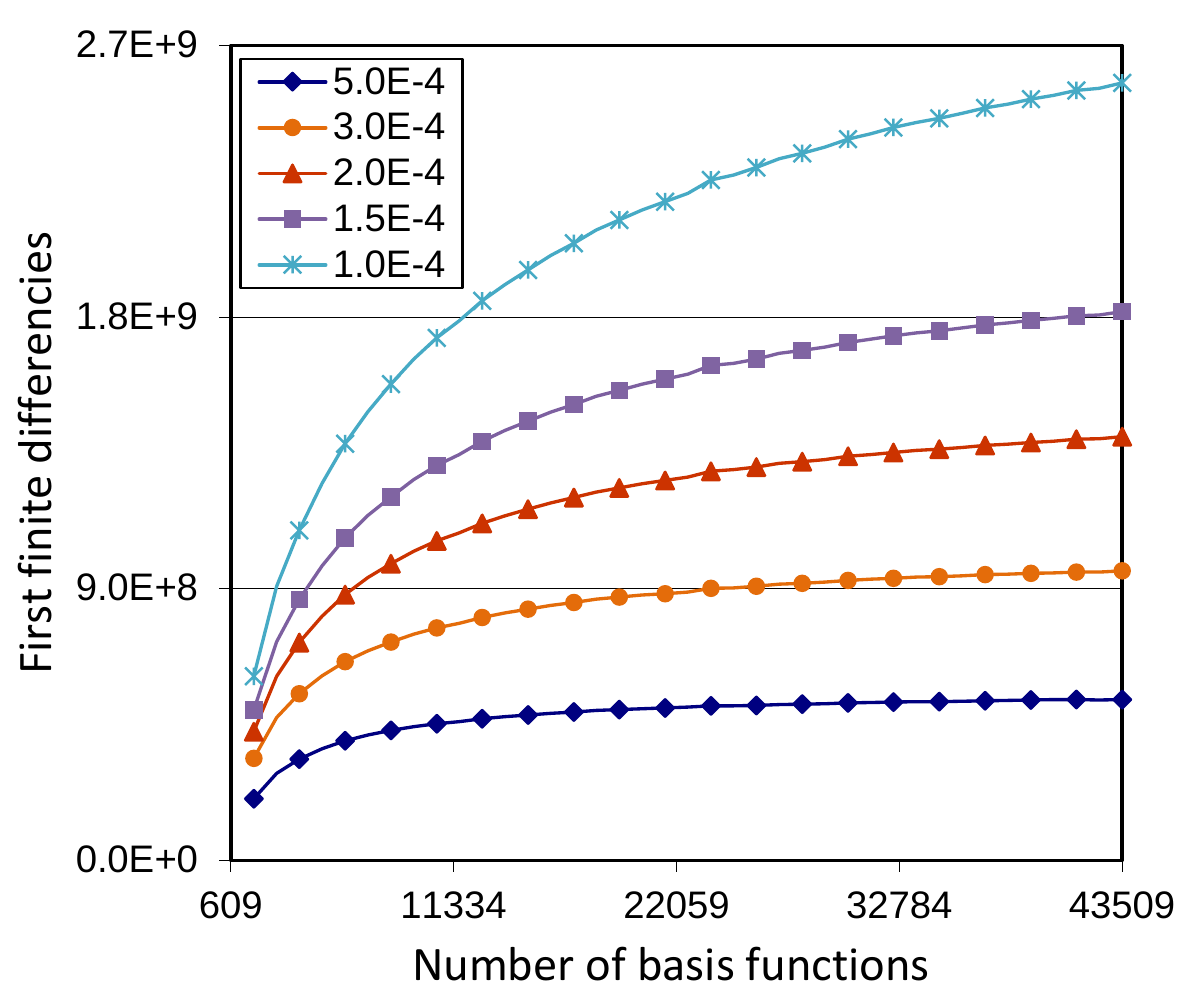}
	} 
	\caption{Dependence of $N_{2e}(\epsilon)$ (a) and first finite differences (b) on $N$ in calculations of alanine polymers with 6-31G basis for integral precision $5\times10^{-4}$, $3\times10^{-4}$, $2\times10^{-4}$, $1.5\times10^{-4}$ and $10^{-4}$.}
	\label{Figure_05_06}
\end{figure}
\begin{figure}
	\centering
	\subfloat[]{\label{fig_07} 
		\includegraphics[scale=0.67]{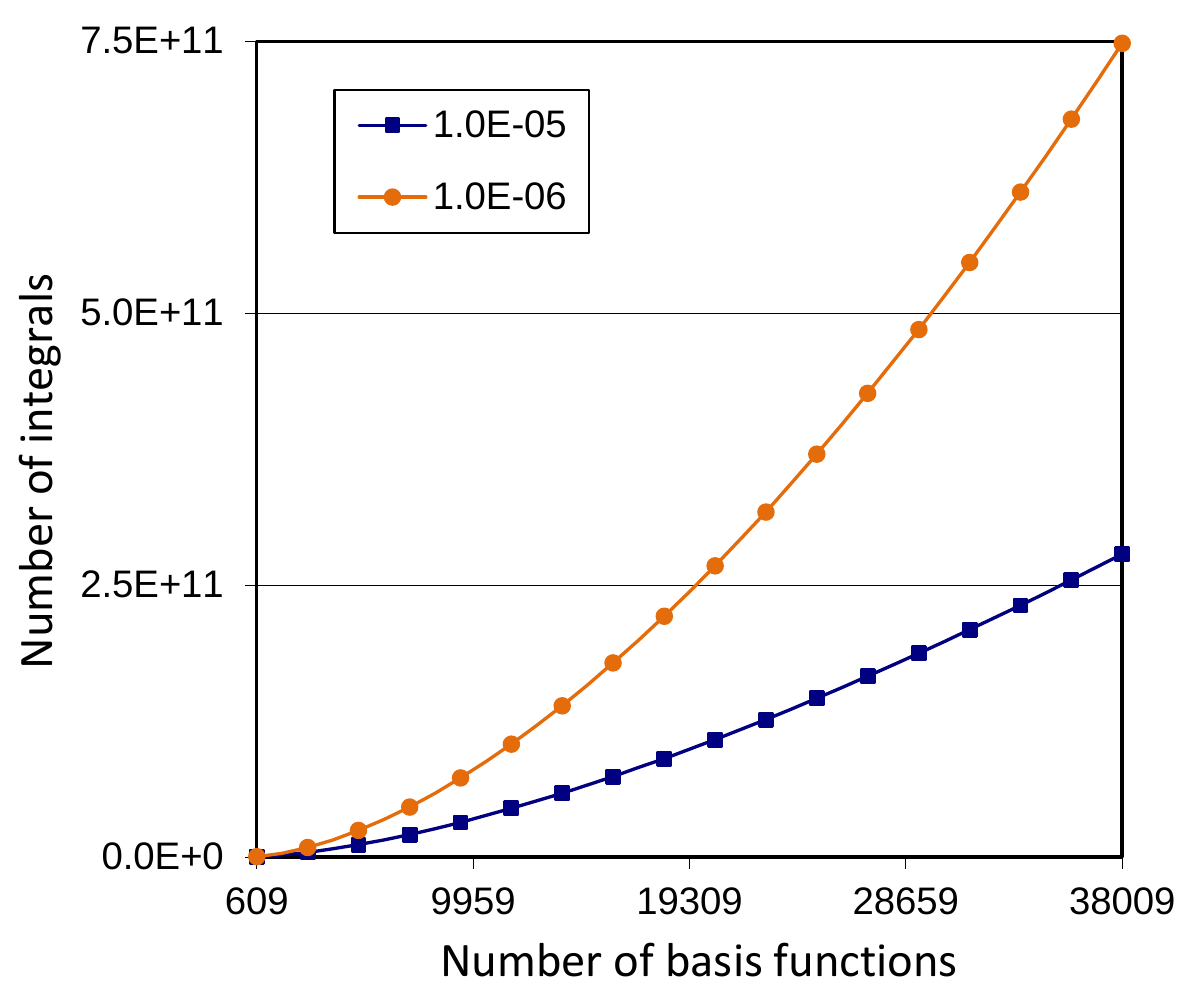}
	} 
    \\
	\subfloat[]{\label{fig_08}
		\includegraphics[scale=0.67]{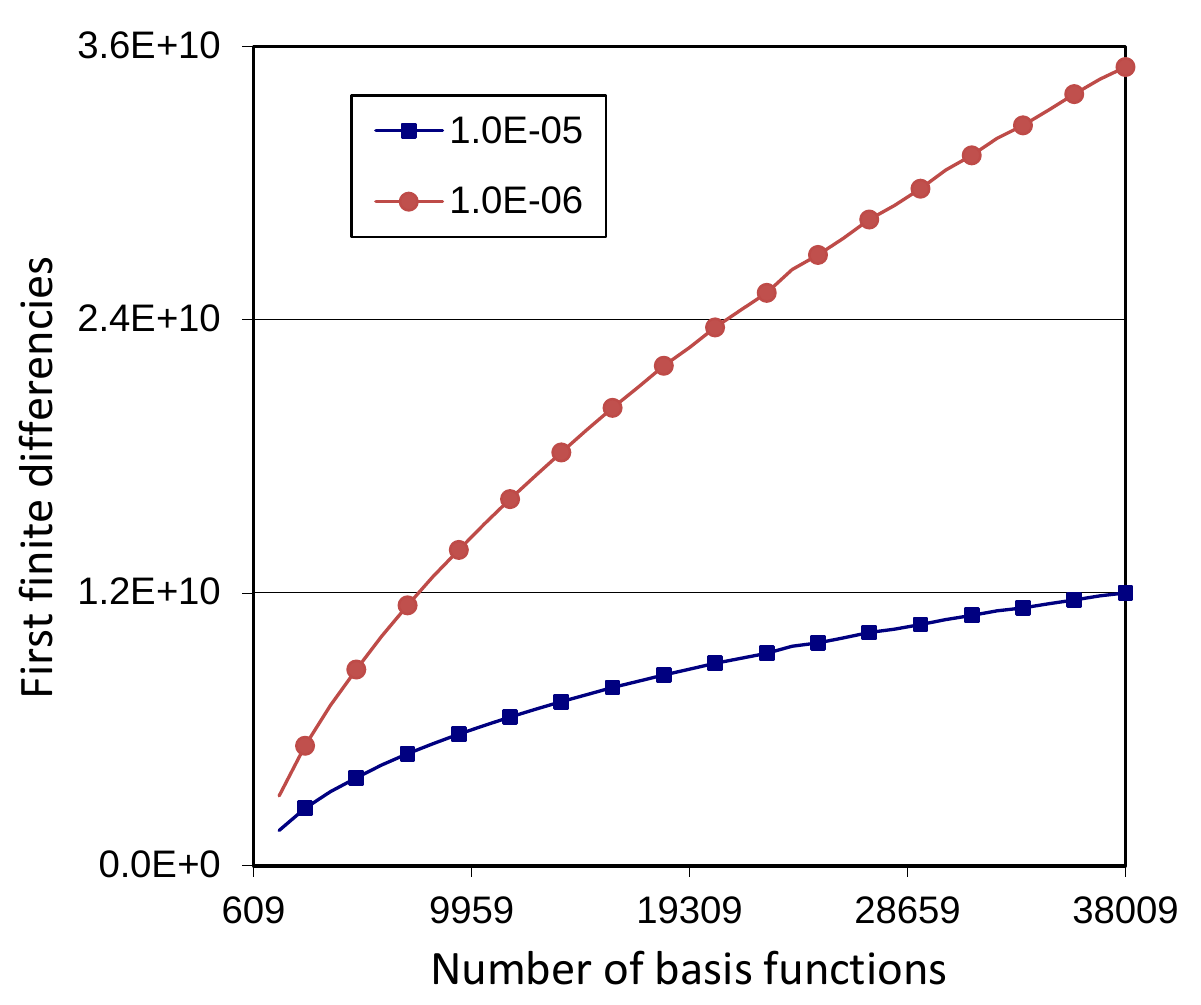}
	} 
	\caption{Dependence of $N_{2e}(\epsilon)$ (a) and first finite differences (b) on $N$ in calculations of alanine polymers with 6-31G basis for integral precision $10^{-5}$ and $10^{-6}$.}
	\label{Figure_07_08}
\end{figure}
\begin{figure}
	\centering
	\subfloat[]{\label{fig_09} 
		\includegraphics[scale=0.67]{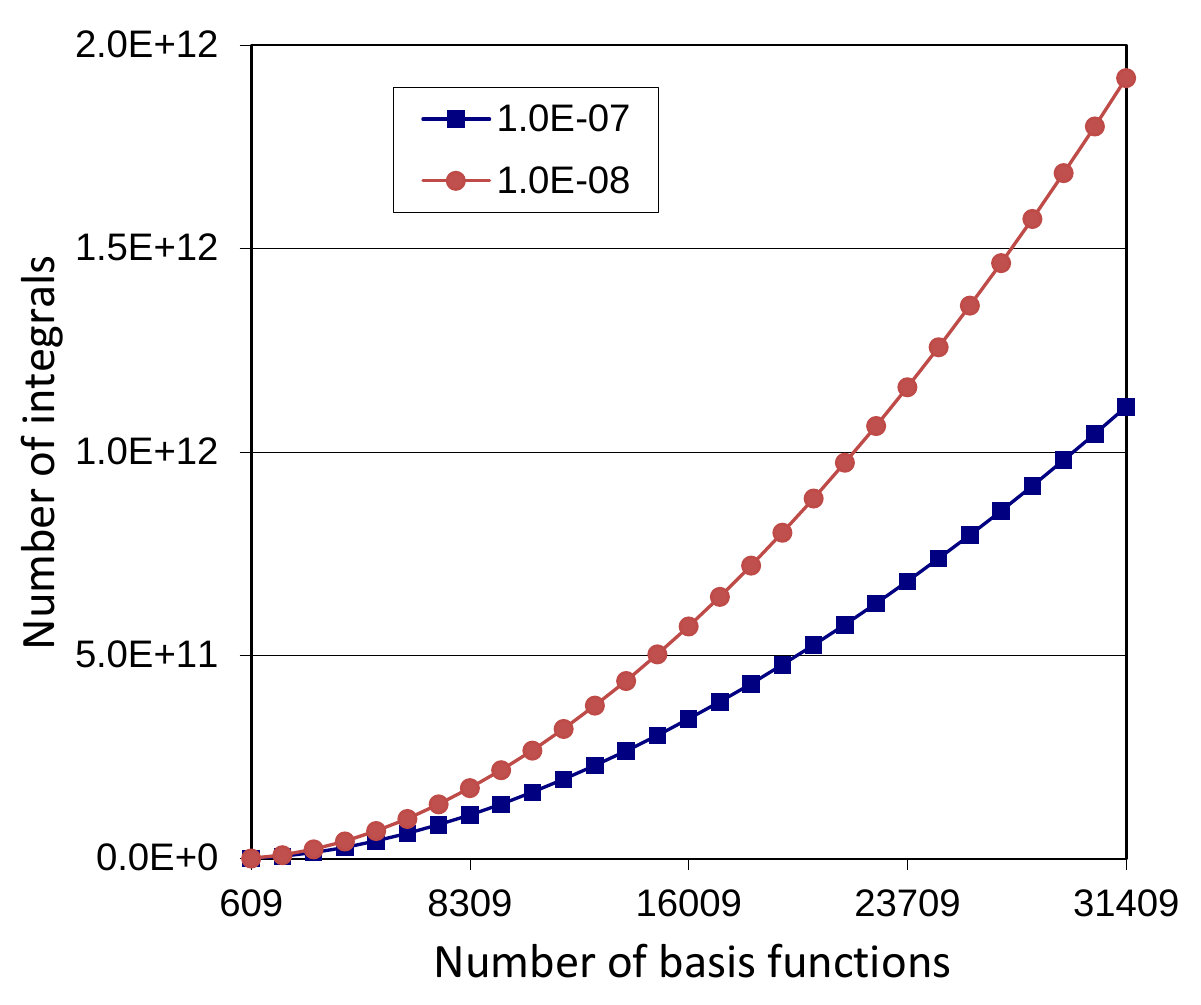}
	} 
    \\
	\subfloat[]{\label{fig_10}
		\includegraphics[scale=0.67]{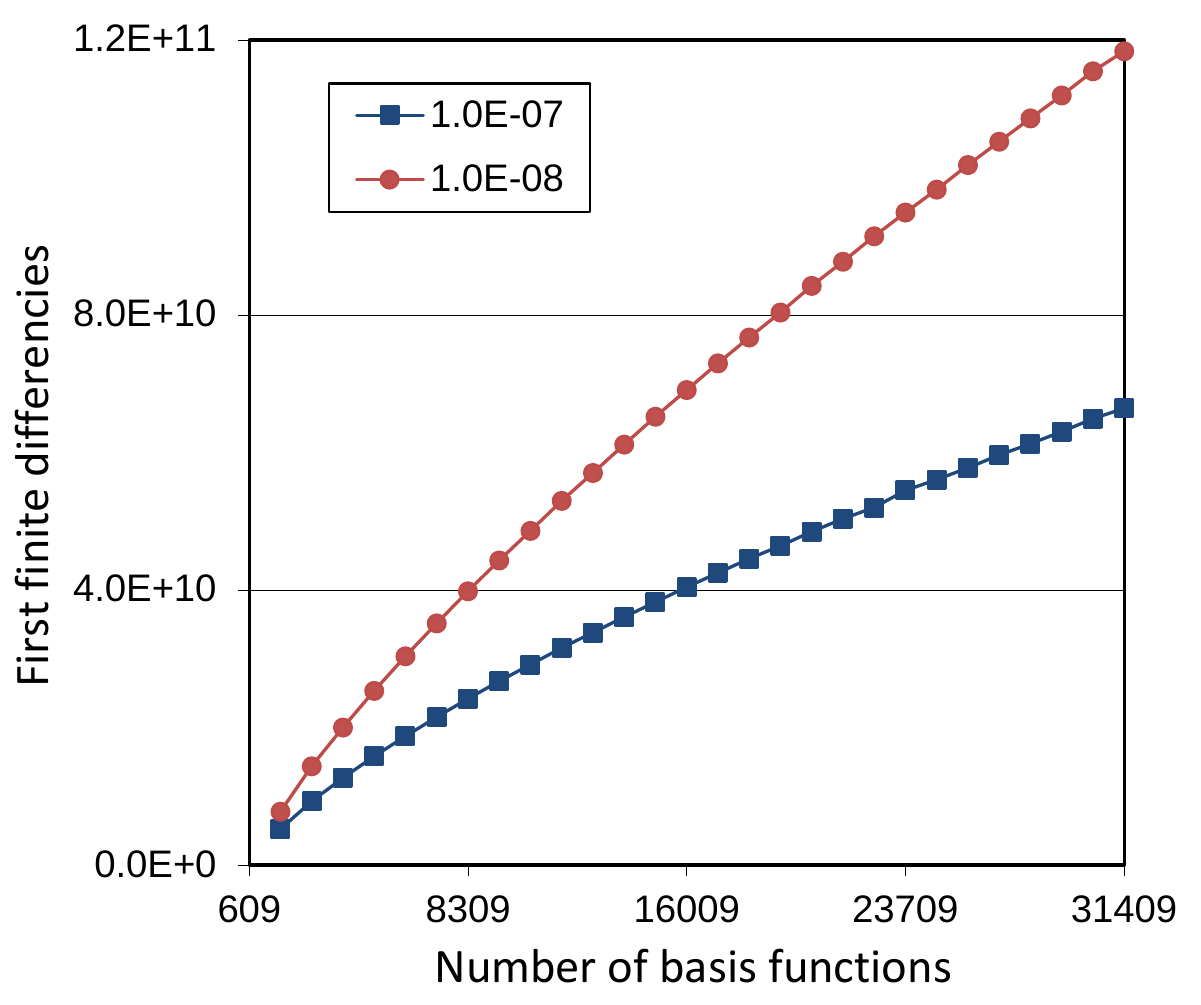}
	} 
	\caption{Dependence of $N_{2e}(\epsilon)$ (a) and first finite differences (b) on $N$ in calculations of alanine polymers with 6-31G basis for integral precision $10^{-7}$ and $10^{-8}$.}
	\label{Figure_09_10}
\end{figure}

An analysis of these results shows that, as expected, the $N_{2e}(\epsilon)$ scales linearly on $N$ for the low cutoff thresholds almost for all $N$, while for higher thresholds the dependencies are still higher than linear. In particular, the $N_{2e}(\epsilon)$ scales linearly for any number of basis functions for integrals with absolute values larger than $10^{-2}$. For integrals with absolute values larger than $10^{-3}$ the $N_{2e}(\epsilon)$ scales linearly for more than 20000 basis functions (see Fig. \ref{Figure_03_04}). The dependencies presented on Fig. \ref{fig_05} and Fig. \ref{fig_06} show that, from the practical point of view, $N_{2e}(\epsilon)$ scales almost linearly for integral precision $2.0\times10^{-4}$ and lower. For higher precision (Fig. \ref{Figure_07_08} and Fig. \ref{Figure_09_10}) the  linear dependence of the number of non-zero two-electron integrals is not observed in the presented calculations. However, the convex form of the dependence of $N_{2e}(\epsilon)$ on $N$ and the concave form of the dependence of the first finite difference derivatives on $N$ points that the power dependence of $N_{2e}(\epsilon)$ on $N$ is lower than the second order. 

Thus, the results of the numerical investigations are in agreement with the theorem which says that a number of non-zero two-electron integrals scales asymptotically linearly with respect to a number of basis functions for large systems.

\section{Linear scaling Fock matrix calculations with density prescreening} \label{Section_3}

There are two possible algorithms for the calculation of the two-electron part of Fock matrix with density prescreening for an evaluation of the contributions of non-zero two-electron integrals. The first one is based on using the total density in a prescreening procedure, while in the second algorithm the Fock matrix is formed from the calculated increments corresponding to the density differences \cite{AFK_JCC_1982_3_385}. In the simplest form, the algorithm with total density prescreening can be presented as follows

\vspace{2.5mm}
\begin{tabular}{lp{6.5cm}} 
	\multicolumn{2}{c}{\bf Algorithm I with total density prescreening} \\
	\textit{Step 1.} & Loop over iterations                    \\
	\textit{Step 2.} & Loops over basis function shells        \\
	\textit{Step 3.} & Estimate maxima of $|(ij|kl)|$ and $|D_{kl}|$ shell blocks \\
	\textit{Step 4.} & If $max|(ij|kl)*D_{kl}|>\epsilon$ calculate \\
	& $F=F+D(2J-K)$ \\
	\textit{Step 5.} & End of loops over basis function shells  \\
	\textit{Step 6.} & Check convergence and, if needed, go to \textit{Step 1.}
\end{tabular}
\vspace{2.5mm}

\noindent{}and the algorithm using density difference prescreening has the following form 

\vspace{2.5mm}
\begin{tabular}{lp{6.5cm}} 
	\multicolumn{2}{c}{\bf Algorithm II with density difference } \\
	\multicolumn{2}{c}{\bf prescreening} \\
	{\it Step 1.} & First iteration. Calculate $F^1=D(2J-K)$ \\ 
	{\it Step 2.} & Loop over iterations. Calculate        \\
	& $\Delta D^m=D^m-D^{m-1}$                 \\
	{\it Step 3.} & Loops over basis function shells        \\
	{\it Step 4.} & Estimate maxima of $|(ij|kl)|$ and $|D_{kl}|$ shell blocks \\
	{\it Step 5.} & If $max|(ij|kl)*\Delta D_{kl}|>\epsilon$ calculate   \\
	& $\Delta F^m=\Delta F^m+\Delta D^m(2J-K)$ \\
	{\it Step 6.} & End of loops over basis function shells  \\
	{\it Step 7.} & Calculate $F^m=F^1+\sum_{i=2}^m \Delta F^m$  \\
	{\it Step 8.} & Check convergence and, if needed, go to {\it Step 2.}
\end{tabular}
\vspace{2.5mm}

The schemes of both algorithms explicitly present only the main steps, which are needed for the selection of two-electron integrals by using prescreening techniques and calculations of Fock and delta Fock matrices. The fast multipole type methods \cite{WhiteHeadGordon_JCP_1994_101_6593, SSF_Science_1996_271_51}, the quantum chemical tree code \cite{CSA_JCP_1996_104_4685}, and the conventional self-consistent-field (SCF) methods with density and density difference prescreenning \cite{MBWP_JCC_2003_24_154} are described by these algorithms. The fast multipole methods and the quantum chemical tree code were realized as direct SCF methods proposed in \cite{Duke_CPL_1974_28_437} and then reintroduced in \cite{AFK_JCC_1982_3_385} together with other improvements.

The {\it Step 4} of Algorithm I and {\it Step 5} of Algorithm II for the Fock and delta Fock matrix calculations are similar and are the most computationally demanding steps of the algorithms. Therefore, they define the scaling property of the Fock and delta Fock matrix calculations with respect to the problem size.  

The scaling property of these steps can be revealed by considering distributions of non-zero two-electron integrals and leading density matrix elements on absolute values in dependence on the problem size. The results presented above show that the non-zero two-electron integrals with large absolute values scales linearly on $N$, although the small integrals scale higher than linearly, however, below quadratic. This means that non-zero two-electron integrals can be considered as a distribution where the largest two-electron integrals or the core part of this distribution scale linearly with respect to the number of basis functions, while the remaining large part of the distribution scales higher than linear. It is also clear that the fractions of the largest two-electron integrals in such a distribution decrease when the number of basis functions increase.

\begin{table} [ht]
	\caption{Dependence of distributions of non-zero two-electron integrals on absolute values in dependence on molecular system size. Two-electron integrals are calculated for alanine polymers $CCOH_3-(NCCOH_2CH_3)_K-NHCH_3$ with 6-31G basis. Normalized fractions of $N_{2e}(\epsilon)$ are given in \%.}
	\label{Table_3}
	\vspace{-1.3mm}
	\begin{ruledtabular}
	\begin{tabular}{@{}rrrrrrr}
		Molec.   &   BF  & 1.0E-4 & 1.0E-5 & 1.0E-6 & 1.0E-7 & 1.0E-8 \\
		\hline 
		\vspace{-0.6mm}
		$Ala_{10}$ &   609 &  11.4  &  13.9  &  19.8  &  25.2  &  29.7  \\
		\vspace{-0.6mm}
		$Ala_{90}$ &  5009 &   5.8  &  10.9  &  19.1  &  27.9  &  36.4  \\ 
		\vspace{-0.6mm}
		$Ala_{170}$ &  9409 &   4.6  &  10.0  &  18.7  &  28.5  &  38.3  \\
		\vspace{-0.6mm}
		$Ala_{250}$ & 13809 &   3.9  &   9.4  &  18.4  &  28.8  &  39.5  \\
		\vspace{-0.6mm}
		$Ala_{330}$ & 18209 &   3.5  &   9.0  &  18.1  &  29.0  &  40.4  \\
		\vspace{-0.6mm}
		$Ala_{410}$ & 22609 &   3.2  &   8.6  &  17.9  &  29.2  &  41.1  \\
		$Ala_{490}$ & 27009 &   2.9  &   8.4  &  17.7  &  29.3  &  41.6  \\
	\end{tabular}
	\end{ruledtabular}
\end{table}

This property of non-zero two-electron integral distributions in calculations of alanine polymers with 6-31G \cite{HDP_JCP_1972_56_2257} basis is presented in Table \ref{Table_3}. The number of basis functions changes from 609 to 27009 in these calculations. The normalized fractions (in \%) of $N_{2e}(\epsilon)$ are given from third to seventh columns. 

Presented results show that the fraction of small integrals increases, while that of large integrals decreases with increase of molecular system size. This property of non-zero two-electron integral distribution follows from the Theorem proven above. Indeed, the considered example corresponds to the case when the number of non-zero large integrals scales linearly on $N$, whereas the number of non-zero small integrals scales higher than linear. For this reason, the fraction of large integrals reduces and the fraction of small integrals increases when increasing the number of basis functions.

It is well known that the number of leading matrix elements of a density matrix linearly scales with respect to the number of basis functions \cite{Kohn_PhysRev_1959_115_809}. On the other hand, the total number of matrix elements of a density matrix is a quadratic function of $N$. Therefore, the distributions of the density matrix elements and the number of non-zero two-electron integrals on absolute value must have similar properties. An example of such a distribution is presented in Table \ref{Table_4}. Density matrices were obtained in Hartree-Fock calculations of alanine polymers with 6-31G basis. The presented results explicitly display that the fraction of the main matrix elements of density matrix reduces with increase of molecular system size. 

\begin{table}[ht]
	\caption{The linear scaling behavior of the leading Hartree-Fock density matrix elements calculated with 6-31G basis for $CCOH_3-(NCCOH_2CH_3)_{K}-NHCH_3$ alanine polymers. Normalized fractions of the leading density matrix elements are given in \%.}
	\label{Table_4}
	\vspace{-1.3mm}
	\begin{ruledtabular}
	\begin{tabular}{@{}crrrrrr}
		Molec.   &  BF  & 1.0E-2 & 1.0E-3 & 1.0E-4 & 1.0E-5 & 1.0E-6 \\
		\hline 
		\vspace{-0.6mm}
		$Ala_{10}$ &  609 &   2.6  &  11.9  &  26.5  &  38.2  &  45.5  \\
		\vspace{-0.6mm}
		$Ala_{20}$ & 1159 &   1.4  &   6.7  &  15.7  &  24.1  &  30.9  \\ 
		\vspace{-0.6mm}
		$Ala_{30}$ & 1709 &   1.0  &   4.6  &  11.1  &  17.2  &  22.5  \\
		\vspace{-0.6mm}
		$Ala_{40}$ & 2259 &   0.7  &   3.6  &   8.6  &  13.5  &  17.8  \\
		\vspace{-0.6mm}
		$Ala_{50}$ & 2809 &   0.6  &   2.9  &   7.0  &  11.0  &  14.6  \\
		$Ala_{60}$ & 3359 &   0.5  &   2.4  &   5.9  &   9.3  &  12.4  \\
	\end{tabular}
	\end{ruledtabular}
\end{table}

The results presented above permit to consider the calculations at \textit{Step 4} and \textit{Step 5} of Algorithms I and II, correspondingly, as a generalized product of two distributions: the distribution of non-zero two-electron integrals and the distribution of density matrix elements. The properties of these distributions are similar, however, they differ in the number of elements. From the analysis of the scaling properties of these distributions follows that, similar to the scaling property of the non-zero two-electron integrals, the Fock and delta Fock matrix calculations at $Step 4$ and $Step 5$ only asymptotically linearly scale with respect to the number of basis functions without using the cutoff criterions. This means that their power dependencies on the number of basis functions are lower than quadratic, but higher than one. 

The use of the cutoff criteria at \textit{Step 4} and \textit{Step 5} of Algorithms I and II enforces the linear scaling property of the Fock matrix calculations. Indeed, the calculations with the largest two-electron integrals and the largest density matrix elements and the calculations with small integrals and small density matrix elements scale differently with respect to the number of basis functions. The calculations with small values give only insignificant contributions to the two-electron part of the Fock or the delta Fock matrices. The use of the cutoff criteria at \textit{Step 4} and \textit{Step 5} excludes these calculations and preserves only those which give important contributions to the Fock or the delta Fock matrices. Thus, the use of the cutoff criteria transforms the asymptotically linear scaling Fock and delta Fock matrix calculations into the linear scaling ones even at those dimensions where the number of non-zero two-electron integrals is still a non-linear scaling function.

A comparison of Algorithm I and Algorithm II shows that one of the main differences between them arises due to different cutoff criteria. At the same time the difference between these criteria is defined by the fact that the density converges to a constant 
$$
\lim_{m \to \infty} max|D| = const ,
$$ 
while the density difference converges to zero 
$$
\lim_{m \to \infty} max|\Delta D| = 0 ,
$$ 
when self-consistent iterations converge. This means that the density difference cutoff criterion is stronger than the density cutoff criterion. Therefore, the scaling property of Algorithm II with density difference prescreening is better than that of Algorithm I. However, one can note that Algorithm I has better numerical stability than Algorithm II, because it is free from the elimination of meaning digits in real numbers which arises due to the difference operation in the latter one. This is the reason for frequently restarting Fock matrix calculations with Algorithm II in ab initio programs.

It needs emphasis that the cutoff criteria at \textit{Step 4} and \textit{Step 5} of Algorithm I and Algorithm II, correspondingly, do not perfectly linearize the Fock and the delta Fock matrix calculations. Therefore, deviations from the linear scaling behavior must be observed in calculations performed with these algorithms. The confirmations of this fact can be found, for example, in publications \cite{SSF_Science_1996_271_51, CSA_JCP_1996_104_4685}. Thus, the Fock matrix calculations with GvFMM method of $C_{6m^2}H_{6m}$, ($m=1$ to $8$) with 3-21G basis scales as the 1.35 power on the number of basis functions \cite{SSF_Science_1996_271_51}. The quantum chemical tree code \cite{CSA_JCP_1996_104_4685} has shown the best power dependencies between 1.9 and 1.6 in calculations of water and $\alpha$-helices with 3-21G basis. The power dependencies displayed by both methods are definitely below quadratic. However, they are not linear, as it must be in exact linear scaling algorithms. 

Thus, the linear scaling property of the Fock matrix calculation with density and density difference prescreening appears due to three factors: the asymptotically linear scaling property of the number of non-zero two-electron integrals, the linear scaling property of leading elements of the density matrix, and the use of the cutoff criterion for linearization of the Fock matrix calculation.

\section{Linear scaling Fock matrix calculations with stored non-zero two-electron integrals, density difference prescreening, and data compression} \label{Section_4}

The above consideration shows that any algorithm of Fock matrix calculation that can be reduced to Algorithms I or II should possesses the linear scaling property. In particular, the conventional algorithm of Fock matrix calculation with stored non-zero two-electron integrals can be reformulated as an algorithm with density (Algorithm I) or density difference (Algorithm II) prescreening \cite{MBWP_JCC_2003_24_154} by using the data compression method \cite{Mitin_Theochem_2002_115} to store non-zero two-electron integrals and their indices. This means that the Fock matrix calculations with stored two-electron integrals modified in such a way should possess the linear scaling property with respect to the number of basis functions.

\begin{figure}
	\centering
	\subfloat[]{\label{fig_11} 
		\includegraphics[scale=0.635]{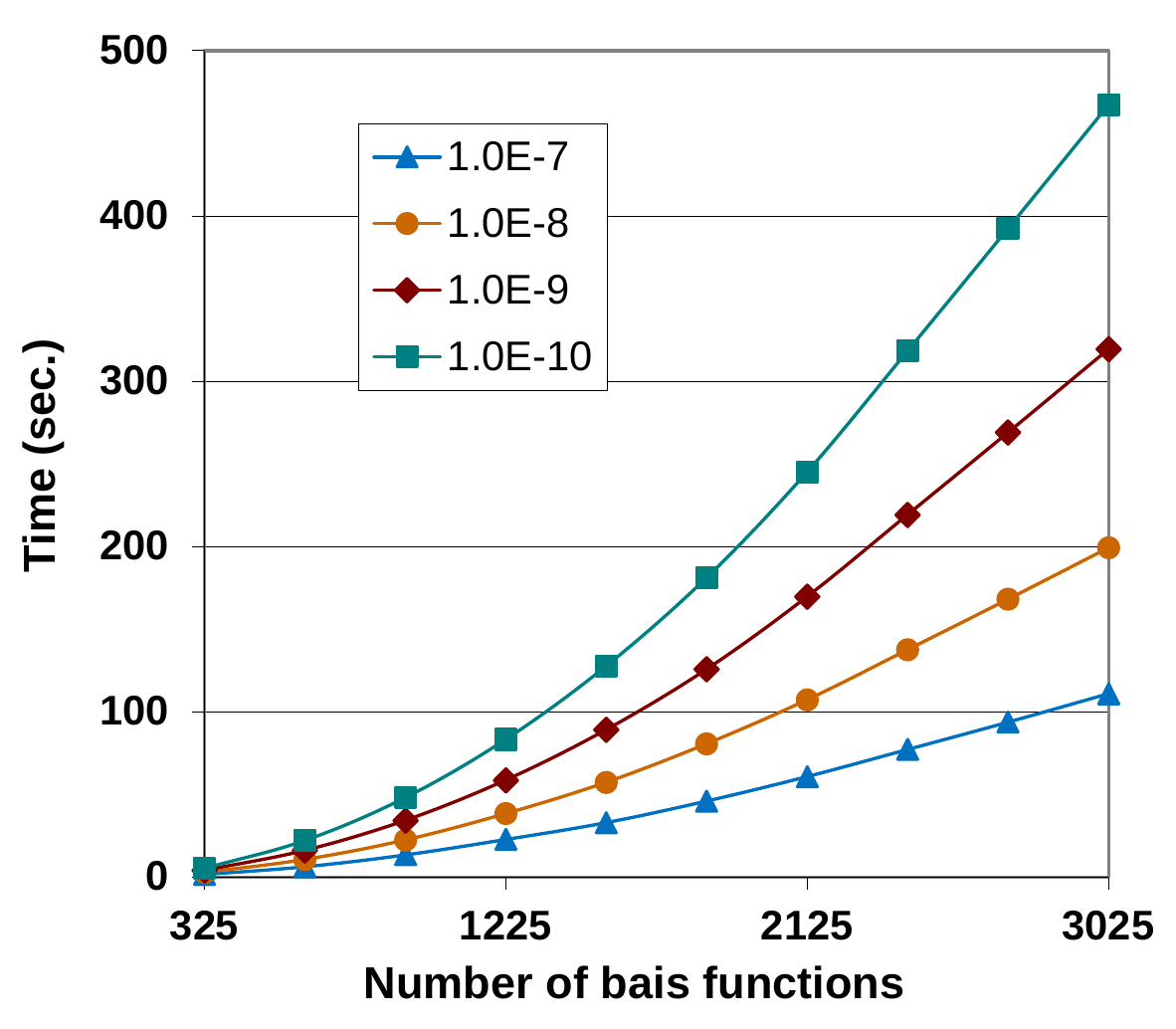}
	} 
	\\
	\subfloat[]{\label{fig_12}
		\includegraphics[scale=0.625]{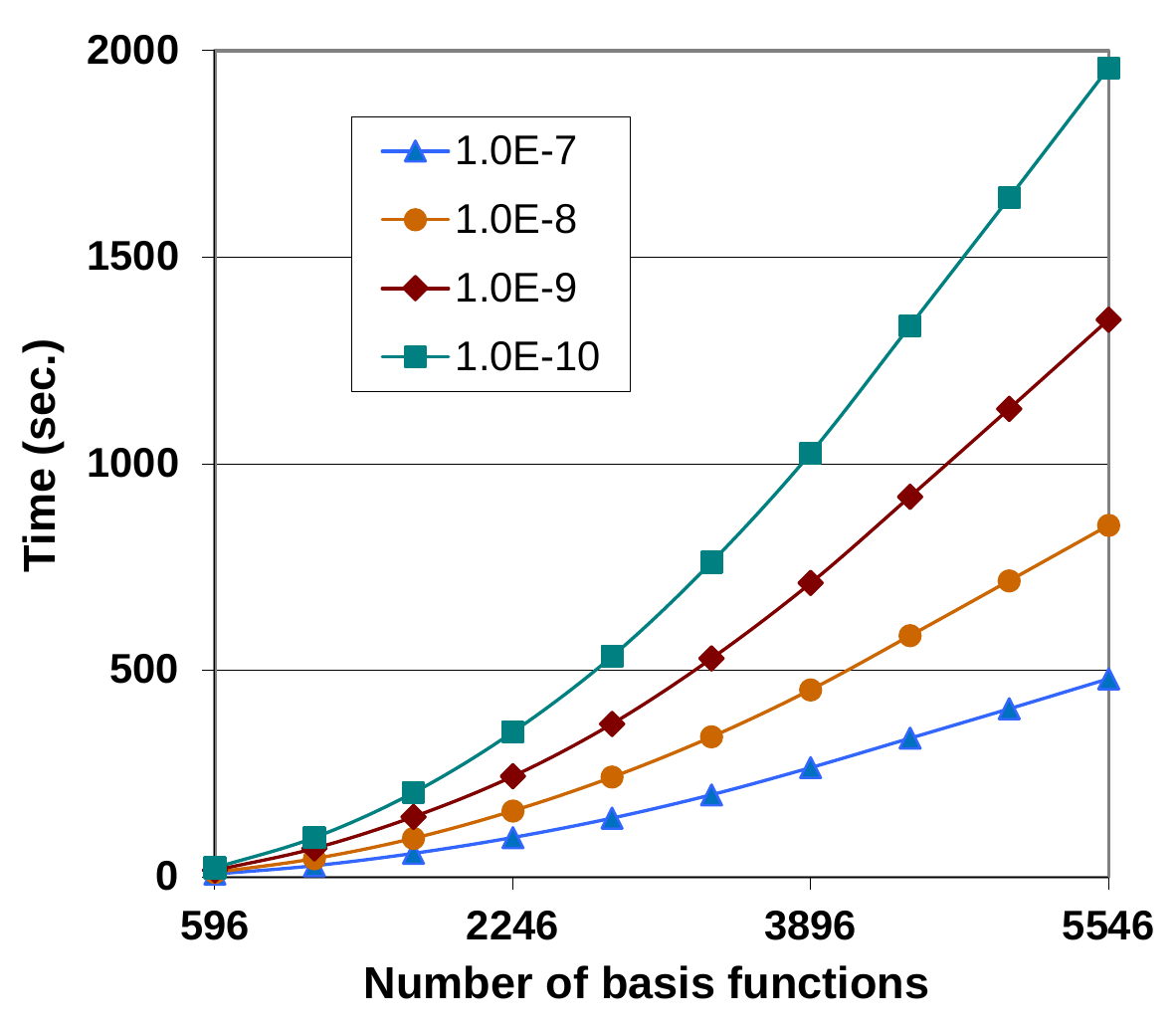}
	} \\
	\subfloat[]{\label{fig_13}
		\includegraphics[scale=0.625]{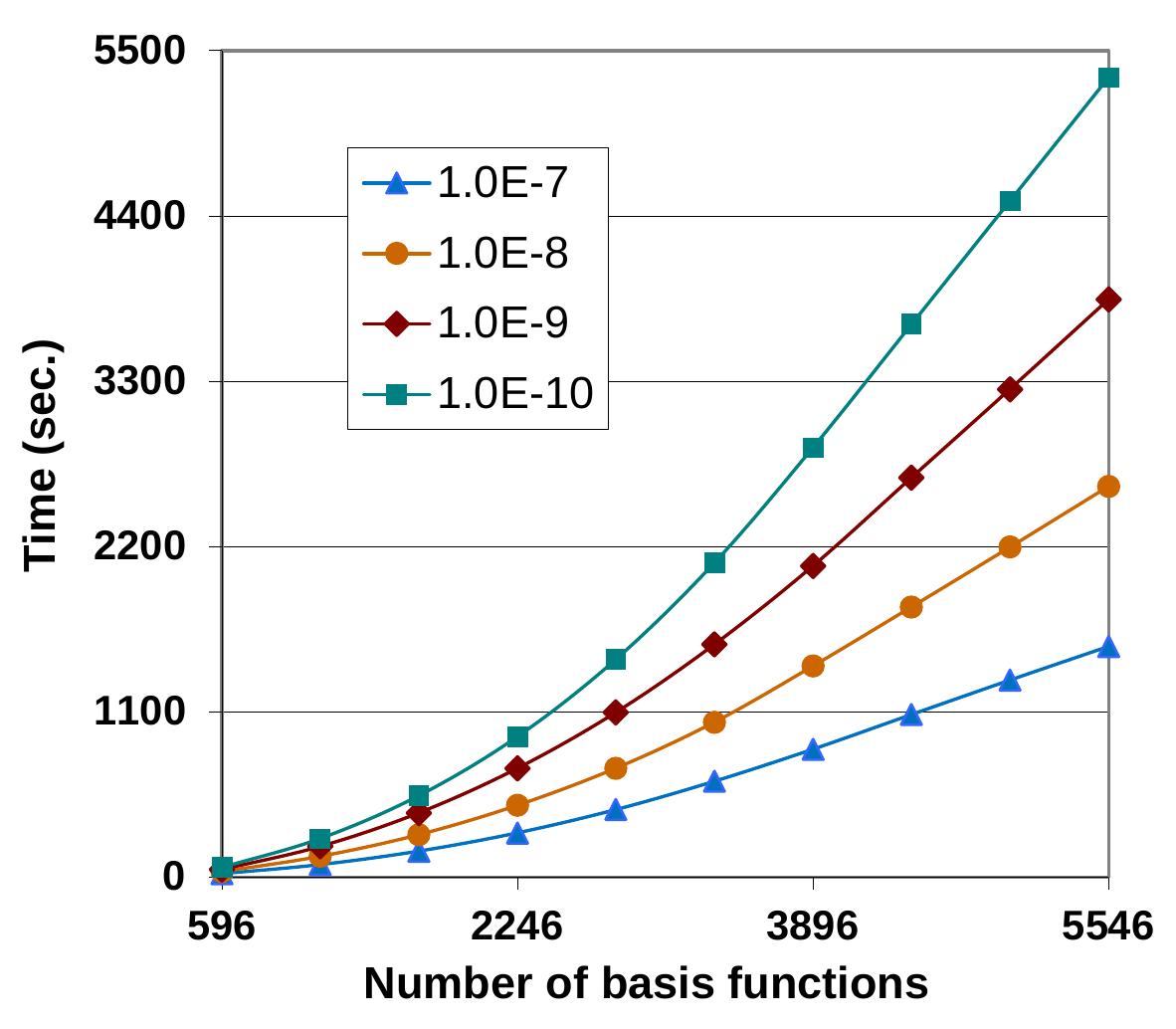}
	} 
	\caption{The time of the Fock matrix calculation for alanine polymers for different cutoff thresholds obtained with basis (a) - STO-3G, (b) - 3-21SP, and (c) - 6-31G.}
	\label{Figure_11_13}
\end{figure}

\begin{figure}
	\centering
	\subfloat[]{\label{fig_14} 
		\includegraphics[scale=0.635]{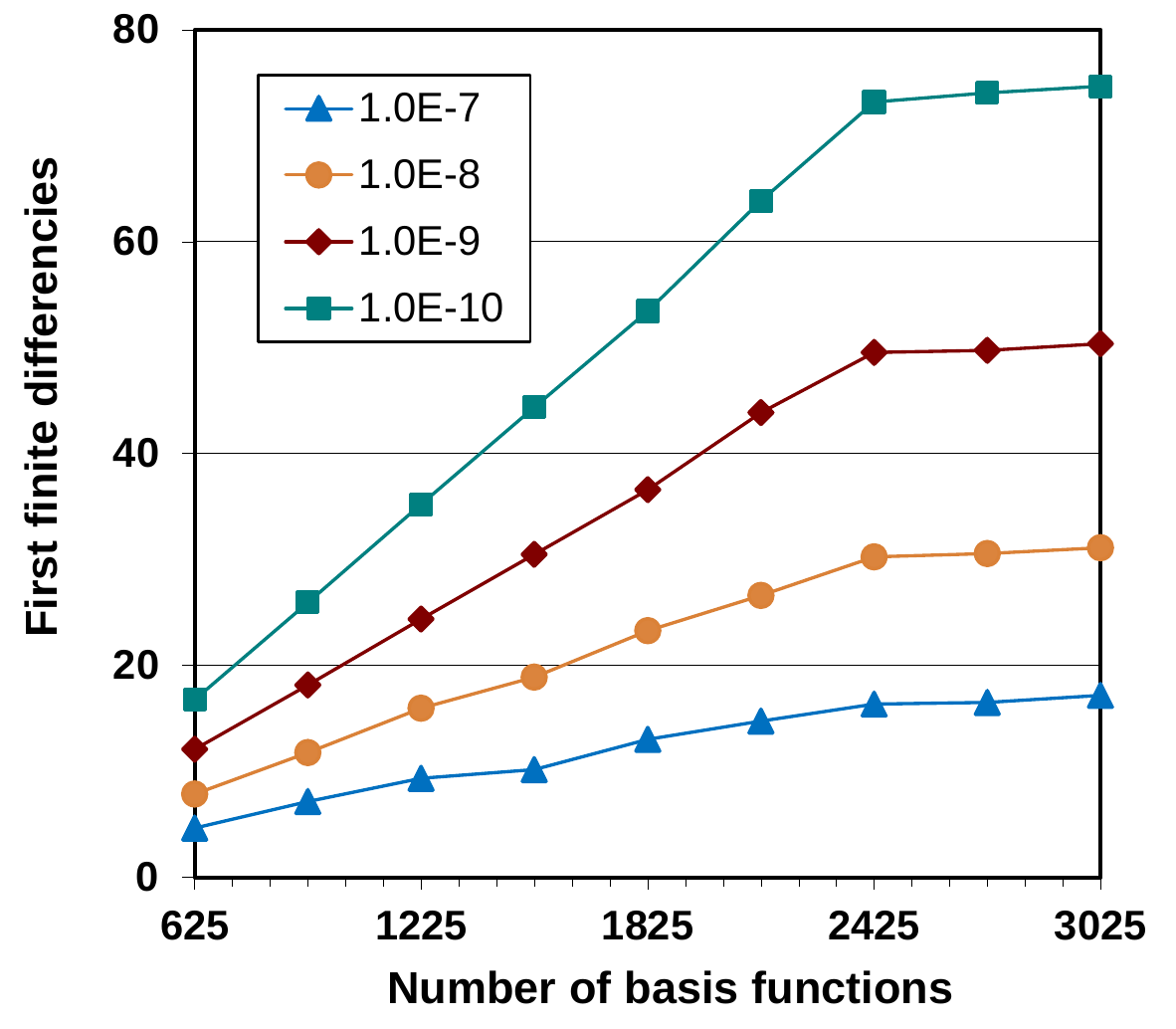}
	}
    \\
	\subfloat[]{\label{fig_15}
		\includegraphics[scale=0.625]{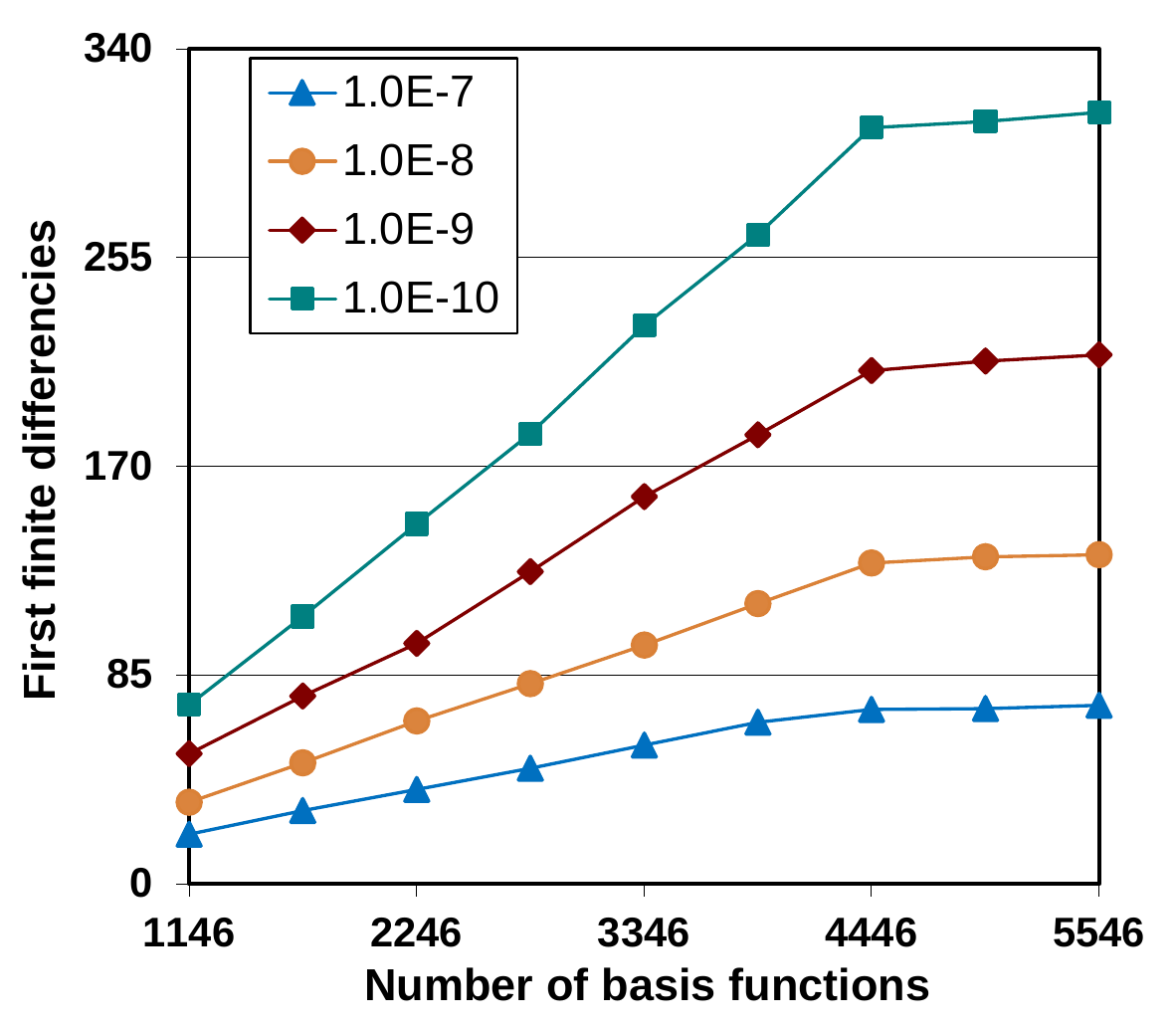}
	} \\
	\subfloat[]{\label{fig_16}
		\includegraphics[scale=0.625]{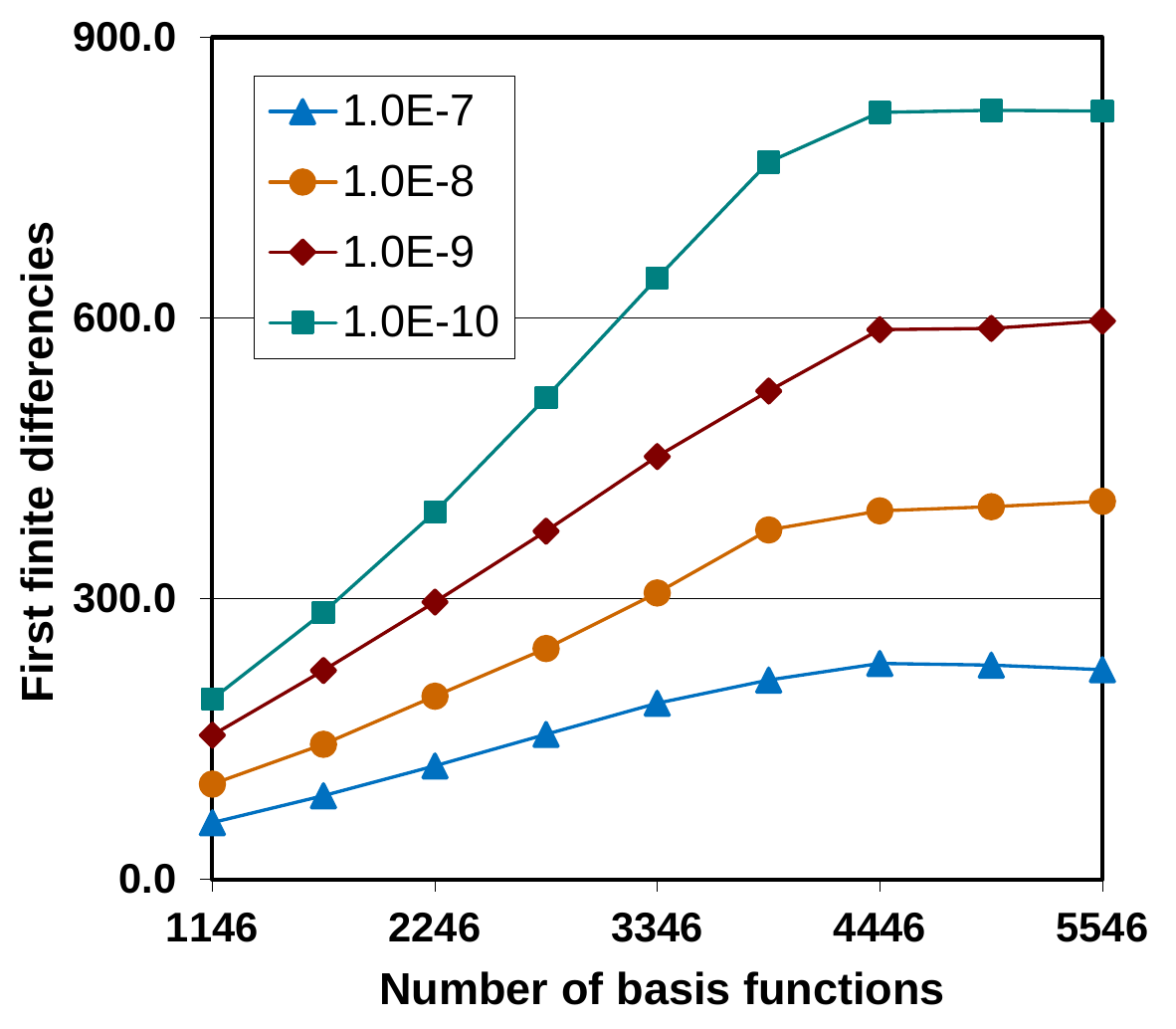}
	} 
	\caption{The first finite differences of the Fock matrix calculation time for different cutoff thresholds obtained with basis (a) - STO-3G, (b) - 3-21SP, and (c) - 6-31G.}
	\label{Figure_14_16}
\end{figure}

The conventional Fock matrix calculation method with stored non-zero two-electron integrals, density difference prescreening and data compression has been realized in program \cite{Mitin_JCC_1998_19_1877} and used for investigations of the scaling property of this method in calculations of alanine polymers $CCOH_3-(NCCOH_2CH_3)_M-NH_2$, when the number of alanine blocks was varied from $M=10$ to $100$. The calculations were performed with STO-3G \cite{HSP_JCP_1969_51_2657}, 3-21SP \cite{MHB_CPL_1996_259_151}, and 6-31G \cite{HDP_JCP_1972_56_2257} basis sets and $\epsilon$ parameter of the prescreening procedure was equal to $10^{-14}$ in all cases. The calculations were performed on PC computer with Intel i7-2600K CPU running at 5.0 GHz. The time of two-electron part Fock matrix calculations in dependence on the number of basis functions for different cutoff precision of two-electron integral calculations are presented on Figures \ref{fig_11} through \ref{fig_13} for basis STO-3G, 3-21SP, and 6-31G correspondingly. The first finite differences for these dependencies, presented on Figures \ref{fig_14} through \ref{fig_16}, were constructed to display their power dependencies.  

A consideration of the numerical derivatives presented on the last figures has to take into account that the density prescreening in the Fock matrix calculation with stored non-zero two-electron integrals is based on using a discrete classification of two-electron integrals on absolute values into five classes \cite{MBWP_JCC_2003_24_154}. Therefore, the time dependencies of the first numerical derivatives obtained with this algorithm and presented on Figure \ref{Figure_14_16} are not as smooth as those which can be obtained with the density difference prescreening algorithm that uses the continuous classification of two-electron integrals on absolute values \cite{AFK_JCC_1982_3_385}.

The dependencies presented in Figures \ref{Figure_11_13} and \ref{Figure_14_16} show that the scaling behavior of the conventional Fock matrix calculation with stored non-zero two-electron integrals, density difference prescreening, and data compression becomes linear at about 2500 basis functions STO-3G basis and about 4000 basis functions for split-valence 3-21SP and 6-31G bases for low precision of two-electron integral calculations. When the two-electron integral cutoff precision increases, the linear scaling regime of the Fock matrix construction shifts to a higher number of basis functions. The first finite difference derivatives in Figures \ref{fig_14}, \ref{fig_15}, and \ref{fig_16}  display this shift especially well.

Thus, numerical tests presented above show that the conventional Fock matrix calculation method with stored non-zero two-electron integrals, density difference prescreening, and data compression possesses the linear scaling property with respect to the number of basis functions. This result is in agreement with the conclusion of the theoretical analysis, given in the previous section, that any algorithm of Fock matrix calculation, which can be reduced to Algorithm I or Algorithm II, has to possess the linear scaling property with respect to the number of basis functions.

\section{Conclusion}

The presented results show that the conventional Fock matrix calculation from the stored non-zero two-electron integrals with density prescreening and data compression possesses the linear scaling property. It follows from the proven Theorem, which says that the total number of non-zero two-electron integrals scales asymptotically linearly with respect to the number of basis functions for large molecular systems. An analysis of the Fock matrix calculations with density and density difference prescreening shows that the linear scaling property of this algorithm arises due to linear scaling properties of the number of non-zero two-electron integrals and the number of leading matrix elements of the density matrix. The density prescreening method reinforces this property, while the density difference prescreening further strengthens of it. An application of the data compression technique to store two-electron integrals and their indices permits the reformulation of the conventional Fock matrix calculations as a method with density or density difference prescreening. Therefore, the Fock matrix calculation from stored non-zero two-electron integrals with density or density difference prescreening and data compression also possesses the linear scaling property with respect to the number of basis functions. The numerical calculations demonstrate that this property begins at 2500 to 4000 basis functions in dependence on the basis function type in molecular calculations.

\bibliography{manuscript_refs}

\end{document}